\documentclass[11pt,a4paper]{article}
\usepackage{amssymb,amsthm,graphicx,color,enumerate,verbatim}
\usepackage{amsmath,framed,cite}
\usepackage{authblk}
\usepackage[hidelinks]{hyperref}
\usepackage{fullpage}
\usepackage{caption}
\usepackage{subcaption}
\usepackage[ruled,vlined,linesnumbered]{algorithm2e}
\usepackage{relsize}
\usepackage{etoolbox}
\providetoggle{long}
\usepackage[none]{hyphenat}
\settoggle{long}{true}

\makeatletter
\DeclareRobustCommand*\cal{\@fontswitch\relax\mathcal}
\makeatother

\usepackage{enumerate}

\DeclareMathOperator{\suf}{suf}

\usepackage[ruled]{algorithm2e} 

\SetAlFnt{\small}
\SetAlCapFnt{\small}
\SetAlCapNameFnt{\small}
\SetAlCapHSkip{0pt}
\IncMargin{-\parindent}

\theoremstyle{plain}
   \newtheorem{lemma}{Lemma}[section]
   \newtheorem{theorem}[lemma]{Theorem}
   \newtheorem{proposition}[lemma]{Proposition}
\theoremstyle{definition}
   
   \newtheorem{example}[lemma]{Example}
\theoremstyle{remark}
   
\numberwithin{equation}{section}

\begin{document}
\title{New Algorithms for Weighted $k$-Domination and\\ Total $k$-Domination Problems in Proper Interval Graphs}

\author[1,2]{Nina Chiarelli\thanks{nina.chiarelli@famnit.upr.si}}
\author[1,2]{Tatiana Romina Hartinger\thanks{tatiana.hartinger@iam.upr.si}}
\author[3,4]{Valeria Alejandra Leoni\thanks{valeoni@fceia.unr.edu.ar}}
\author[3,5]{\newline Maria In\'es Lopez Pujato\thanks{mineslpk@hotmail.com}}
\author[1,2]{Martin Milani\v c\thanks{martin.milanic@upr.si}}
\affil[1]{\normalsize University of Primorska, Faculty of Mathematics, Natural Sciences and Information Technologies, Glagolja\v ska 8, SI-6000 Koper, Slovenia}
\affil[2]{\normalsize University of Primorska, Andrej Maru\v si\v c Institute, Muzejski trg 2, SI-6000 Koper, Slovenia}
\affil[3]{\normalsize FCEIA, Universidad Nacional de Rosario, Rosario, Santa Fe, Argentina}
\affil[4]{\normalsize CONICET, Argentina}
\affil[5]{\normalsize ANPCyT, Argentina}


\date{\today}

\maketitle

\begin{abstract}
Given a positive integer $k$, a $k$-dominating set in a graph $G$ is a set of vertices such that every vertex not in the set has at least $k$ neighbors in the set. A total $k$-dominating set, also known as a $k$-tuple total dominating set, is a set of vertices such that every vertex of the graph has at least $k$ neighbors in the set. The problems of finding the minimum size of a $k$-dominating, respectively total $k$-dominating set, in a given graph, are referred to as $k$-domination, respectively total $k$-domination.  These generalizations of the classical domination and total domination problems are known to be NP-hard in the class of chordal graphs, and, more specifically, even in the classes of split graphs (both problems) and undirected path graphs (in the case of total $k$-domination). On the other hand, it follows from recent work of Kang et al.~(2017) that these two families of problems are solvable in time $\mathcal{O}(|V(G)|^{6k+4})$ in the class of interval graphs. We develop faster algorithms for $k$-domination and total $k$-domination in the class of proper interval graphs, by means of reduction to a single shortest path computation in a derived directed acyclic graph with $\mathcal{O}(|V(G)|^{2k})$ nodes and
$\mathcal{O}(|V(G)|^{4k})$ arcs. We show that a suitable implementation, which avoids constructing all arcs of the digraph, leads to
a running time of
$\mathcal{O}(|V(G)|^{3k})$. The algorithms are also applicable to the weighted case.
\end{abstract}

\noindent
{\bf Keywords:} $k$-domination, total $k$-domination, proper interval graph, polynomial-time algorithm

\section{Introduction}

Variants of domination in graphs form a rich area of graph theory, with many useful and interesting concepts, results, and challenging problems~\cite{MR1605684,MR1605685,MR3060714}. In this paper we consider a family of generalizations of classical domination and total domination known as $k$-domination and total $k$-domination. Given a graph $G$ and a positive integer $k$, a \emph{$k$-dominating set} in $G$ is a set $S\subseteq V(G)$ such that every vertex $v\in V(G)\setminus S$ has at least $k$ neighbors in $S$, and
a \emph{total $k$-dominating set} in $G$ is a set $S\subseteq V(G)$ such that every vertex $v\in V(G)$ has at least $k$ neighbors in $S$.
The $k$-domination and the total $k$-domination problems aim to find the minimum size of a $k$-dominating, resp.~total $k$-dominating set, in a given graph. The notion of $k$-domination was introduced by Fink and Jacobson in 1985~\cite{MR812671} and studied in a series of papers (e.g.,~\cite{MR3043099,MR3043104,MR2974046,MR2776174,MR2370182}) and in a survey Chellali et al.~\cite{MR2863534}. The notion of total $k$-domination was introduced by Kulli in 1991~\cite{KULLI1991+} and studied under the name of \emph{$k$-tuple total domination} by Henning and Kazemi in 2010~\cite{MR2607047} and also in a series of recent papers~\cite{MR2974027,MR2960327,MR3215459,ALT2016+}. The terminology ``$k$-tuple total domination'' was introduced in analogy with the notion of ``$k$-tuple domination'', introduced in 2000 by Harary and Haynes~\cite{MR1755232}.\footnote{A set $S$ of vertices is said to be a \emph{$k$-tuple dominating set} if every vertex of $G$ is adjacent or equal to at least $k$ vertices in $S$.} The redundancy involved in $k$-domination and total $k$-domination problems makes them useful in various applications, for example in forming sets of representatives or in resource allocation in distributed computing systems (see, e.g.,~\cite{MR1605684}). However, these problems are known to be NP-hard~\cite{MR995851,MR2960327} and also hard to approximate~\cite{MR3027964}.

The $k$-domination and total $k$-domination problems are NP-hard not only for general graphs but also in the class of chordal graphs. More specifically, the problems are NP-hard in the class of split graphs~\cite{MR3043104,MR2960327} and, in the case of total $k$-domination, also in the class of undirected path graphs~\cite{MR3215459}. We consider  $k$-domination and total $k$-domination in another subclass of chordal graphs, the class of proper interval graphs. A graph $G$ is an \emph{interval graph} if it has an intersection model consisting of closed intervals on a real line, that is, if there exist a family ${\cal I}$ of intervals on the real line and a one-to-one correspondence between the vertices of $G$ and the intervals of ${\cal I}$ such that two vertices are joined by an edge in $G$ if and only if the corresponding intervals intersect. A \emph{proper interval graph} is an interval graph that has a \emph{proper interval model}, that is, an intersection model in which no interval contains another one. Proper interval graphs were introduced by Roberts~\cite{MR0252267}, where it was shown that they coincide with the unit interval graphs, that is, interval graphs having an intersection model consisting of intervals of unit length. Various characterizations of proper interval graphs have been developed in the literature (see, e.g.,~\cite{MR2364171,MR1180196,MR1368745,MR1203643}) and several linear-time recognition algorithms are known, which in case of a yes instance also compute a proper interval model (see, e.g.,~\cite{MR2049655} and references cited therein).

The usual domination and total domination problems (that is, when $k = 1$) are known to be solvable in linear time in the class of interval graphs (see~\cite{MR1143909,MR1622646,MR1610005} and~\cite{MR845445,MR930198,MR942582,MR1622646,MR971614}, respectively).
Furthermore, for each fixed integer $k\ge 1$, the $k$-domination and total $k$-domination problems are solvable in time $\mathcal{O}(n^{6k+4})$ in the class of interval graphs where $n$ is the order of the input graph. This follows from recent results due to Kang et al.~\cite{Kang2017arXiv}, building on previous works by Bui-Xuan et al.~\cite{MR3126918} and Belmonte and Vatshelle~\cite{MR3126917}. In fact, Kang et al.~studied a more general class of problems, called $(\rho,\sigma)$-domination problems, and showed that
every such problem can be solved in time $\mathcal{O}(n^{6d+4})$
in the class of $n$-vertex interval graphs, where $d$ is a parameter associated to the problem (see Corollary 3.2 in~\cite{Kang2017arXiv} and the paragraph following it). The value of parameter $d$ for $k$-domination and total $k$-domination equals $k$, yielding the claimed time complexity.

\subsection{
Our Results and Approach}

\begin{sloppypar}
To the best of our knowledge, the only known polynomial-time algorithms for $k$-domination and total $k$-domination for a general (fixed) $k$ in the class of interval graphs follow from the above-mentioned work of Kang et al.~\cite{Kang2017arXiv} and run in time $\mathcal{O}(n^{6k+4})$. We significantly improve the above result for the case of proper interval graphs. We show that for each positive integer $k$, the $k$-domination and total $k$-domination problems are solvable in time $\mathcal{O}(n^{3k})$ in the class of proper interval graphs. Except for $k = 1$, this significantly improves on the best known running time.

Our approach is based on a reduction showing that for each positive integer $k$, the total $k$-domination problem on a given proper interval graph $G$ can be reduced to a shortest path computation in a derived edge-weighted directed acyclic graph. A similar reduction works for $k$-domination. The reductions immediately result in algorithms with running time $\mathcal{O}(n^{4k+1})$. We then show that with a suitable implementation the running time can be improved to $\mathcal{O}(n^{3k})$. The algorithms can be easily adapted to the weighted case, at no expense in the running time.
\end{sloppypar}

An extended abstract of this work appeared in the Proceedings of ISCO 2018~\cite{ISCO2018}.

\subsection{
Related Work}

We now give an overview of related work and compare our results with the most related other results, in addition to those due to Kang et al.~\cite{Kang2017arXiv}, which motivated this work.

\begin{sloppypar}
\medskip
\noindent{\bf Overview.}
Several results on the complexity of $k$-domination and  total
$k$-domination problems were established in the literature. For every $k$, the $k$-domination problem is NP-hard in the classes of bipartite graphs~\cite{BFH2017} and split graphs~\cite{MR3043104}. The problem is solvable in linear time in the class of graphs every block of which is a clique, a cycle or a complete bipartite graph (including trees, block graphs, cacti, and block-cactus graphs)~\cite{MR3043104}, and, more generally, in any class of graphs of bounded clique-width~\cite{MR2232389,MR1739644} (see also~\cite{MR3197779}).
For \hbox{every} \hbox{positive} integer $k$, the total $k$-domination problem is NP-hard in the classes of split graphs~\cite{MR2960327}, doubly chordal graphs~\cite{MR2960327}, bipartite graphs~\cite{MR2960327}, undirected path graphs~\cite{MR3215459}, and, for $k\in \{2,3\}$, also in the class of bipartite planar graphs~\cite{ALT2016+}. The problem is solvable in linear time in the class of graphs every block of which is a clique, a cycle, or a complete bipartite graph~\cite{MR3215459}, and, more generally, in any class of graphs of bounded clique-width~\cite{MR2232389,MR1739644}, and in polynomial time in the class of chordal bipartite graphs~\cite{MR2960327}.
$k$-domination and total $k$-domination problems were also studied with respect to their (in)approximability properties, both in general~\cite{MR3027964} and in restricted graph classes~\cite{BFH2017}, as well as from the parameterized complexity point of view~\cite{MR3556042,MR3193864}.
\end{sloppypar}


\begin{sloppypar}
Besides $k$-domination and total $k$-domination, other variants of domination problems solvable in polynomial time in the class of proper interval graphs (or in some of its superclasses) include $k$-tuple domination for all $k\ge 1$~\cite{MR1979859} (see also~\cite{MR2427750} and, for $k = 2$,~\cite{MR2861126}), connected domination~\cite{MR942582}, independent domination~\cite{MR687354}, paired domination~\cite{MR2354752}, efficient domination~\cite{MR1362658}, liar's domination~\cite{MR3095468}, restrained domination~\cite{MR3349893}, eternal domination~\cite{MR3327095}, power domination~\cite{MR3010415}, outer-connected domination~\cite{MR3385035}, Roman domination~\cite{MR2467312}, Grundy domination~\cite{MR3586754}, etc.
\end{sloppypar}

\medskip
\noindent{\bf Comparison.} Bertossi~\cite{MR873860} showed how to reduce the total domination problem in a given interval graph to a shortest path computation in a derived edge-weighted directed acyclic graph satisfying some additional constraints on pairs of consecutive arcs. A further transformation reduces the problem to a usual (unconstrained) shortest path computation. Compared to the approach of Bertossi, our approach exploits the additional structure of {\it proper} interval graphs in order to gain generality in the problem space. Our approach works for every $k$ and is also more direct, in the sense that the (usual or total, unweighted or weighted) $k$-domination problem in a given proper interval graph is reduced to a shortest path computation in a derived edge-weighted directed acyclic graph in a single, unified step.

The works of Liao and Chang~\cite{MR1979859} and of Lee and Chang~\cite{MR2427750} consider various domination problems in the class of strongly chordal graphs (and, in the case of~\cite{MR1979859}, also dually chordal graphs). While the class of strongly chordal graphs generalizes the class of interval graphs, the domination problems studied in~\cite{MR1979859,MR2427750} all deal with closed neighborhoods, and for those cases structural properties of strongly chordal and dually chordal graphs are helpful for the design of linear-time algorithms. In contrast, $k$-domination and total $k$-domination are defined via open neighborhoods and results of~\cite{MR1979859,MR2427750} do not seem to be applicable or easily adaptable to our setting.

\bigskip
\noindent{\bf Structure of the paper.}
In Section~\ref{sec:tkd-construction}, we describe the reduction for the total $k$-domination problem. The specifics of the implementation resulting in improved running time are given in Section~\ref{sec:time}.
In Section~\ref{sec:kd}, we discuss how the approach can be modified to solve the $k$-domination problem. Extensions to the weighted cases are presented in Section~\ref{sec:kd-weighted}. We conclude the paper with some open problems in Section~\ref{sec:conclusion}.

\medskip
In the rest of the section, we fix some definitions and notation. Given a graph $G$ and a set $X\subseteq V(G)$, we denote by $G[X]$ the subgraph of $G$ induced by $X$ and by $G-X$ the subgraph induced by $V(G)\setminus X$. For a vertex $u$ in a graph $G$, we denote by $N(u)$ the set of neighbors of $u$ in $G$.
Note that for every graph $G$, the set $V(G)$ is a $k$-dominating set, while $G$ has a total $k$-dominating set if and only if
every vertex of $G$ has at least $k$ neighbors.
For notions not defined here, we refer the reader to~\cite{MR1367739,MR2572804}.

\section{The Reduction for Total $k$-Domination}\label{sec:tkd-construction}

Let $k$ be a positive integer and $G = (V,E)$ a given proper interval graph. We may assume that $G$ is equipped with a proper interval model ${\cal I} = \{I_j\mid j=1,\ldots,n\}$ where $I_j=\left[a_j,b_j\right]$ for all $j = 1,\ldots, n$. We may also assume that no two intervals coincide. Moreover, since in a proper interval model the order of the left endpoints equals the order of the right endpoints, we assume that the intervals are sorted increasingly according to their left endpoints, i.e., $a_1<\ldots<a_n$.
We use notation $I_j<I_\ell$ if $j<\ell$ and say in this case that $I_j$ is \emph{to the left of} $I_\ell$ and $I_\ell$ is \emph{to the right of} $I_j$.
Also, we write $I_j\le I_\ell$ if $j\le \ell$. Given three intervals $I_j, I_\ell, I_m\in {\cal I}$, we say that interval $I_\ell$ is \emph{between} intervals $I_j$ and $I_m$ if $j<\ell<m$. We say that interval $I_j$ \emph{intersects} interval $I_\ell$ if $I_j\cap I_\ell\neq \emptyset$.

Our approach can be described as follows. Given $G$, we compute an edge-weighted directed acyclic graph $D^t_{k}$ (where the superscript ``$t$'' means ``total'' and $k$ is the constant specifying the problem) and show that the total $k$-domination problem on $G$ can be reduced to a shortest path computation in $D^t_{k}$. In what follows, we first give the
definition of digraph $D^t_{k}$ and illustrate the construction on an example (Example~\ref{example}). Next we explain the intuition behind the reduction and, finally, prove the correctness of the reduction.

\medskip
To distinguish the vertices of $D^t_{k}$ from those of $G$, we refer to them as \emph{nodes}. Vertices of $G$ are typically denoted by $u$ or $v$, and nodes of $D^t_{k}$ by $s, s',s''$.
Each node of $D^t_{k}$ is a sequence of intervals from the set ${\cal I}'= {\cal I}\cup \{I_0,I_{n+1}\}$,
where $I_0$, $I_{n+1}$ are two new, ``dummy'' intervals
such that $I_0< I_1$, $I_0\cap I_1 = \emptyset$, $I_{n}<I_{n+1}$, and $I_n\cap I_{n+1} = \emptyset$.
We naturally extend the linear order $<$ on ${\cal I}$ to the whole set ${\cal I}'$. We say that an interval $I\in {\cal I}'$ is \emph{associated with} a node $s$ of $D^t_{k}$ if it appears in sequence $s$. Given a node $s$ of $D^t_{k}$, we denote the set of all intervals associated with $s$ by ${\cal I}_s$. The first and the last interval in ${\cal I}_s$ with respect to ordering $<$ of ${\cal I}'$ are denoted by $\min(s)$ and~$\max(s)$, respectively. A sequence $(I_{i_1},\ldots, I_{i_q})$ of intervals from ${\cal I}$ is said to be \emph{increasing} if $i_1<\ldots <i_q$.

The node set of $D^t_{k}$ is given by $V(D^t_{k}) = \{I_0,I_{n+1}\}\cup S\cup B$, where:
\begin{itemize}
\item $I_0$ and $I_{n+1}$ are sequences of intervals of length one.\footnote{This assures that the intervals $\min(s)$ and $\max(s)$ are well defined also for $s\in \{I_0, I_{n+1}\}$, in which case both are equal to~$s$.}

\smallskip
\item $S$ is the set of so-called \emph{small nodes}. Set $S$  consists exactly of those
increasing sequences $s = (I_{i_1},\ldots, I_{i_q})$ of (not necessarily consecutive) intervals from ${\cal I}$ such that:
\begin{enumerate}[(1)]
  {\setlength\itemindent{0.2cm}\item $k+1\le q \le 2k-1$,}
  {\setlength\itemindent{0.2cm}\item $I_{i_j}\cap I_{i_{j+1}}\neq\emptyset$ for all $j\in \{1,\ldots, q-1\}$, and}
  {\setlength\itemindent{0.2cm}\item every interval $I\in {\cal I}$  such that $\min(s)\le I\le \max(s)$ intersects at least $k$ intervals from the set ${\cal I}_s\setminus\{I\}$.}
\end{enumerate}

\smallskip
\item $B$ is the set of so-called \emph{big nodes}. Set $B$ consists exactly of those increasing sequences $s = (I_{i_1},\ldots, I_{i_{2k}})$ of (not necessarily consecutive) intervals from ${\cal I}$ of length $2k$ such that:
\begin{enumerate}[(1)]
  {\setlength\itemindent{0.2cm}\item $I_{i_j}\cap I_{i_{j+1}}\neq\emptyset$ for all $j\in \{1,\ldots, 2k-1\}$ and
  }{\setlength\itemindent{0.2cm}\item every interval $I\in {\cal I}$  such that $I_{i_k}\le I \le I_{i_{k+1}}$ intersects at least $k$ intervals from the set
  ${\cal I}_s\setminus\{I\}$.}
\end{enumerate}
\end{itemize}

The arc set of $D^t_{k}$ is given by $E(D^t_{k}) = E_0\cup E_1$, where:
\begin{itemize}
\item Set $E_0$ consists exactly of those ordered pairs $(s,s')\in V(D^t_{k})\times V(D^t_{k})$ such that:
\begin{enumerate}[\indent(1)]
  {\setlength\itemindent{0.2cm}\item $\max(s)< \min(s')$ and $\max(s)\cap \min(s') = \emptyset$,
  }{\setlength\itemindent{0.2cm}\item every interval $I\in {\cal I}$ such that $\max(s)<I<\min(s')$ intersects at least $k$ intervals from
  ${\cal I}_s\cup {\cal I}_{s'}$,
  }{\setlength\itemindent{0.2cm}\item if $s$ is a big node, then the rightmost $k+1$ intervals associated with $s$ pairwise intersect, and
  }{\setlength\itemindent{0.2cm}\item if $s'$ is a big node, then the leftmost $k+1$ intervals associated with $s'$ pairwise intersect.}
\end{enumerate}

\smallskip
\item Set $E_1$ consists exactly of those ordered pairs $(s,s')\in V(D^t_{k})\times V(D^t_{k})$ such that
$s,s'\in B$ and
there exist $2k+1$ intervals $I_{i_1},\ldots, I_{i_{2k+1}}$ in ${\cal I}$ such that
$s = (I_{i_1},I_{i_2},\ldots, I_{i_{2k}})$ and $s' = (I_{i_2},I_{i_3},\ldots, I_{i_{2k+1}})$.
\end{itemize}

To every arc $(s,s')$ of $D^t_{k}$ we associate a non-negative length $\ell(s,s')$, defined as follows:
\begin{equation}\label{eq:length}\tag{$\ast$}
\ell(s,s')=\left \{
\begin{array}{ll}
|{\cal I}_{s'}|, & \hbox{if $(s,s')\in E_0$ and $s'\neq I_{n+1}$;}\\
1, & \hbox{if $(s,s')\in E_1$;}\\
0, & \hbox{otherwise.}
\end{array}
\right.
\end{equation}
The length of a directed path in $D^t_{k}$ is defined, as usual, as the sum of the lengths of its arcs.

\begin{example}\label{example}
\begin{sloppypar}
Consider the problem of finding a minimum total $2$-dominating set in the graph $G$ given by the proper interval model ${\cal I}$ depicted in Figure~\ref{fig:example}$(a)$. Using the reduction described above, we obtain the digraph $D^t_{2}$ depicted in Figure~\ref{fig:example}$(c)$ along with the length function on arcs, where, for clarity, nodes $(I_{i_1},\ldots, I_{i_p})$ of $D^t_{2}$ are identified with the corresponding strings of indices $i_1i_2\ldots i_p$. We also omit in the figure the (irrelevant) nodes that do not belong to any directed path from $I_0$ to $I_{n+1}$.
There is a unique shortest $I_0,I_{9}$-path in $D^t_{2}$, namely $(0,2356,3567,9)$. The path corresponds to  $\{2,3,5,6,7\}$, the only minimum total $2$-dominating set in $G$. 

\begin{figure}[h!]
  \centering
  \includegraphics[width=0.91\textwidth]{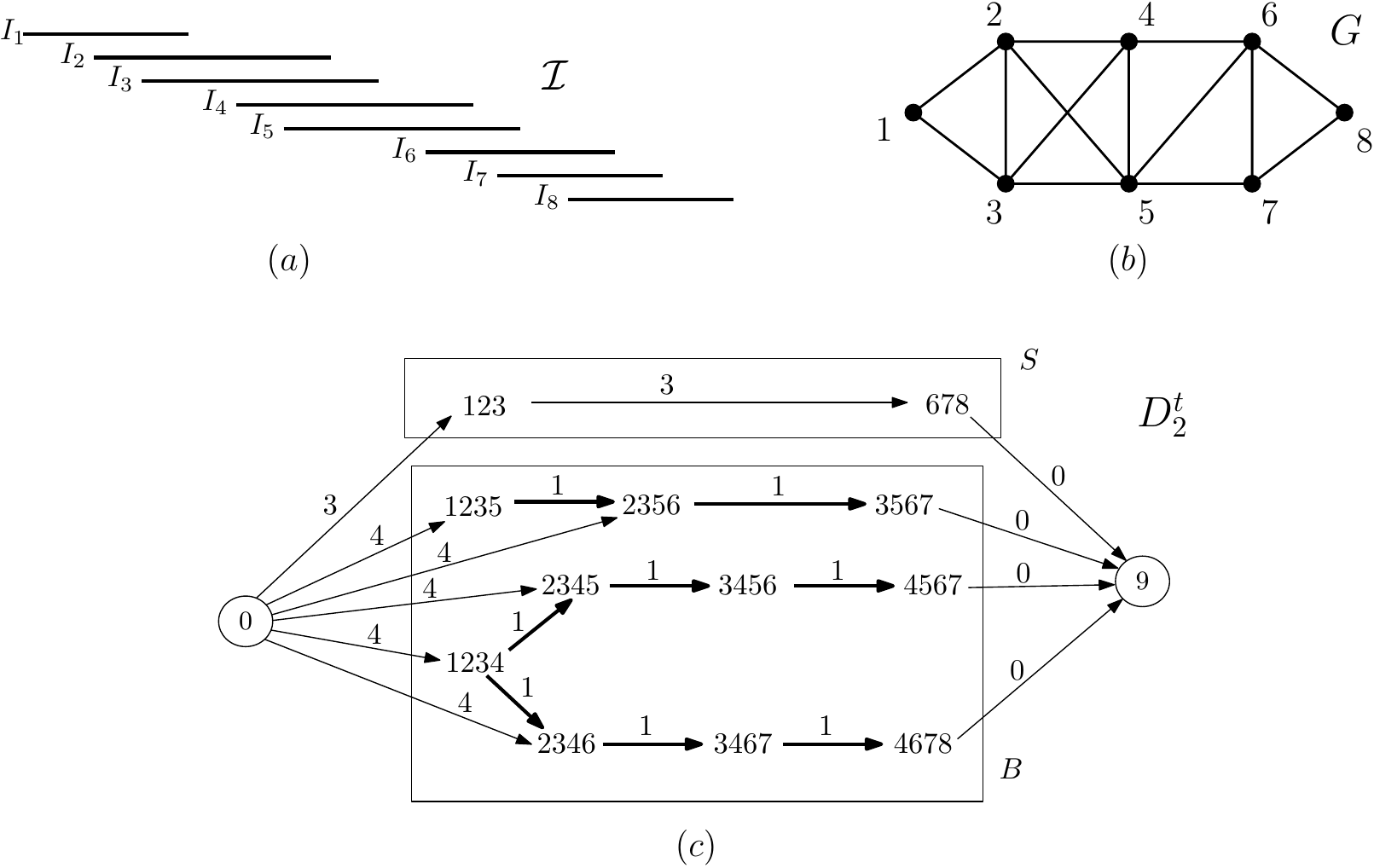}\\
  \caption{$(a)$ A proper interval model ${\cal I}$, $(b)$ the corresponding proper interval graph $G$, and $(c)$ a part of the derived digraph $D^t_{2}$, where only nodes that lie on some directed path from $I_0$ to $I_{9}$ are shown.
    Edges in $E_1$ are depicted bold.}\label{fig:example}
\end{figure}
\end{sloppypar}
\end{example}

\smallskip

\begin{sloppypar}
\begin{lemma}\label{lem:components}
Given a graph $G$ and a positive integer $k$, let $U$ be a total $k$-dominating set in $G$, and let $C$ be a component of $G[U]$. Then $|V(C)|\ge k+1$.
\end{lemma}
\end{sloppypar}

\begin{proof}
Let $u\in V(C)$. Since $U$ is a total $k$-dominating set, $u$ has at least $k$ neighbors in $U$.
However, since $C$ is a component of $G[U]$, all the neighbors of $u$ in $U$ are in fact in $C$.
Thus, $|V(C)|\ge |\{u\}\cup (N(u)\cap U)|\ge k+1$.
\end{proof}

The following proposition
establishes the correctness of the reduction.

\begin{proposition}\label{prop:correctness}
Given a proper interval graph $G$ and a positive integer $k$, let $D^t_{k}$ be the directed graph constructed as above.
Then $G$ has a total $k$-dominating set of size $c$ if and only if $D^t_{k}$ has a directed path from $I_0$ to $I_{n+1}$ of length $c$.
\end{proposition}

Before giving a detailed proof of Proposition~\ref{prop:correctness}, we explain the intuition behind the reduction. The subgraph of $G$ induced by a minimum total $k$-dominating set may contain several connected components. These components as well as vertices within them are naturally ordered from left to right. Moreover, since each connected subgraph of a proper interval graph has a Hamiltonian path, the nodes of $D^t_k$
correspond to paths in $G$, see condition (2) for small nodes or condition (1) for big nodes.
Since each vertex of $G$ has at least $k$ neighbors in the total $k$-dominating set, each component has at least $k+1$ vertices.
Components with at least $2k$ vertices give rise to directed paths in $D^t_k$ consisting of big nodes and arcs in $E_1$.
Each component with less than $2k$ vertices corresponds to a unique small node in $D^t_k$, which can be seen as a trivial directed path in $D^t_k$. The resulting paths inherit the left-to-right ordering from the components. Any two consecutive paths (with respect to this ordering) are joined in $D^t_k$ by an arc in $E_0$. Moreover, $I_0$ is joined to the leftmost node of the leftmost path by an arc in $E_0$ and, symmetrically,
the rightmost node of the rightmost path is joined to $I_{n+1}$ by an arc in $E_0$. Adding such arcs yields a directed path from $I_0$ to $I_{n+1}$ of the desired length.

The above process can be reversed. Given a directed path $P$ in $D^t_k$ from $I_0$ to $I_{n+1}$, a total $k$-dominating set in $G$ of the desired size can be obtained as the set of all vertices corresponding to intervals in ${\cal I}$ associated with the internal nodes of $P$. Note that our construction of the graph $D^t_k$ implies that such a set is indeed a total $k$-dominating set in $G$. For example, condition (3) from the definition of arcs in $E_0$
guarantees that the vertex corresponding to the rightmost interval associated with $s\in B$ (where $(s,s')\in E_0$) is $k$-dominated. The condition is related to the fact that in proper interval graphs the neighborhood of a vertex represented by an interval $[a,b]$ splits into two cliques: one for all intervals containing $a$ and another one for all intervals containing $b$.

\medskip

The digraph $D_k^t$ has $\mathcal{O}(n^{2k})$ nodes and $\mathcal{O}(n^{4k})$ arcs and can be, together with the length function $\ell$ on its arcs, computed from $G$ directly from the definition in time $\mathcal{O}(n^{4k+1})$. A shortest directed path (with respect to $\ell$) from $I_0$ to all nodes reachable from $I_0$ in $D^t_{k}$ can be computed in polynomial time using any of the standard approaches, for example using Dijkstra's algorithm. Actually, since $D^t_{k}$ is acyclic, a dynamic programming approach along a topological ordering of $D^t_{k}$ can be used to compute shortest paths from $I_0$ in linear time (in the size of  $D^t_{k}$).
Proposition~\ref{prop:correctness} therefore implies that the total $k$-domination problem is solvable in time $\mathcal{O}(n^{4k+1})$
in the class of $n$-vertex proper interval graphs.

We will show in Section~\ref{sec:time} that, with a careful implementation, a shortest $I_0, I_{n+1}$-path in $D^t_{k}$ can be computed without examining all the arcs of the digraph, leading to an improved running time of $\mathcal{O}(n^{3k})$.

\subsection*{Proof of Proposition~\ref{prop:correctness}}

We assume all notation up to Example~\ref{example} and, additionally, fix the following notation useful for both directions of the proof: for $X\subseteq V(G)$, we denote by ${\cal I}^X$ the set of intervals in ${\cal I}$ corresponding to vertices in $X$.

\medskip
First we establish the forward implication. Suppose that $G$ has a total $k$-dominating set $U$ of size $c$.
The components of $G[U]$ can be naturally ordered from left to right, say as
$C_1,\ldots, C_r$. To each component $C_i$ we will associate a path $P^i$ in $D^t_{k}$ defined as a sequence of nodes.
The desired directed path $P$ from $I_0$ to $I_{n+1}$ in $D^t_{k}$ will be then obtained
by combining the paths $P^i$ into a single sequence of nodes preceded by $I_0$ and followed by $I_{n+1}$.

We say that a component $C_i$ is \emph{small} if $|V(C_i)|<2k$ and
\emph{big}, otherwise. To every small component $C_i$ we associate a sequence $s^i=(I_{i_{1}}, \ldots, I_{i_{p}})$, consisting of the $p=|V(C_i)|$ intervals corresponding to the vertices of $C_i$, ordered increasingly. We claim that $s^i$ is a small node of $D^t_{k}$. By Lemma~\ref{lem:components}, $C_i$ has at least $k+1$ vertices, which, together with the fact that $C_i$ is small, implies property (1) of the definition of a small node. Property (2) follows from the fact that ${\cal I}$ is a proper interval model of $G$ and $C_i$ is connected. To show that $s^i$ satisfies property (3) of the definition of a small node, consider an interval $I \in {\cal I}$ such that $\min(s^i) \le I \le \max(s^i)$. Let $u$ be the vertex of $G$ corresponding to $I$. Since $U$ is a total $k$-dominating set, vertex $u$ has at least $k$ neighbors in $U$. Since $\min(s^i) \le I \le \max(s^i)$, all the neighbors of $u$ in $U$ must belong to $C_i$, more specifically, $N(u) \cap U \subseteq V(C_i) \setminus \{u\}$. It follows that $I$ intersects at least $k$ intervals from the set ${\cal I}_{s^i} \setminus \{I\}$, establishing also property (3) and with it the claim that $s^i$ is a small node of $D^t_{k}$.
The path $P^i$ associated to component $C_i$ is the one-node path having $s^i$ as a node.

Let now $C_i$ be a big component and let $p = |V(C_i)|$. Then $p\ge 2k$.
Let $I_{j_1},\ldots,  I_{j_p}$ be the intervals corresponding to the vertices of $C_i$, ordered increasingly.
For every $q \in \{1,\ldots, p-2k+1\}$, let $s^i_q$ denote the subsequence of these intervals of length $2k$ starting at $q$-th interval,
that is, $s^i_q = (I_{j_q}, I_{j_{q+1}},\ldots, I_{j_{q+2k-1}})$. We claim that for each $q\in \{1,\ldots, p-2k+1\}$, sequence $s^{i}_{q}$ is a big node of $D^t_{k}$. Property (1) follows from the fact that ${\cal I}$ is a proper interval model of $G$ and $C_i$ is connected.  To show that $s^i_q$ satisfies property (2) of the definition of a big node, consider an interval $I \in {\cal I}$ such that $I_{j_{q+k-1}} \le I \le I_{j_{q+k}}$ and let $u$ be the vertex of $G$ corresponding to $I$. Since $U$ is a total $k$-dominating set, vertex $u$ has at least $k$ neighbors in $U$. Since $I_{j_{1}} \le I \le I_{j_{p}}$, all the neighbors of $u$ in $U$ must belong to $C_i$. It follows that $I$ intersects at least $k$ intervals from the set $\{I_{j_{1}}, \ldots, I_{j_{p}}\}\setminus\{I\}$. Suppose for a contradiction that $I$ intersects strictly less that $k$ intervals from the set ${\cal I}_{s^i_q}\setminus\{I\}=\{I_{j_{q}}, I_{j_{q+1}}, \ldots, I_{j_{q+2k-1}}\}\setminus\{I\}$. Since $I$ intersects at least $k$ intervals from the set $\{I_{j_{1}}, \ldots, I_{j_{p}}\}\setminus\{I\}$, there is an interval, call it $I'$, in the set $\{I_{j_{1}}, \ldots, I_{j_{p}}\} \setminus {\cal I}_{s^i_q}$ such that $I \cap I' \neq \emptyset$. Note that
$\{I_{j_{1}}, \ldots, I_{j_{p}}\} \setminus {\cal I}_{s^i_q} = {\cal I}_1\cup {\cal I}_2$ where
${\cal I}_1 = \{I_{j_1}, \ldots, I_{j_{q-1}}\}$ and ${\cal I}_2 = \{I_{j_{q+2k}}, \ldots, I_{j_{p}}\}$. Suppose first that $I'\in {\cal I}_1$, that is, $I' = I_{j_{\alpha}}$ for some $\alpha\in \{1,\ldots, q-1\}$.
Since ${\cal I}$ is a proper interval model of $G$, conditions $I_{j_{q+k-1}}\le I\le I_{j_{q+k}}$ and $I_{j_{q+k-1}}\cap I_{j_{q+k}}\neq \emptyset$ imply that interval $I$ intersects each of the two intervals $I_{j_{q+k-1}}$ and $I_{j_{q+k}}$ (possibly $I \in \{I_{j_{q+k-1}}, I_{j_{q+k}}\}$). Similarly, the fact that $I$ intersects both $I'=I_{j_{\alpha}}$ and  $I_{j_{q+k}}$ implies that
$I$ also intersects each of the intervals in the set $\{I_{j_{\alpha}}, I_{j_{\alpha+1}},\ldots, I_{j_{q+k}}\}$; in particular, $I$ intersects each of the $k+1$ intervals in the set $\{I_{j_{q}}, I_{j_{q+1}},\ldots, I_{j_{q+k}}\}$, all of which are in ${\cal I}_{s^i_q}$. It follows that interval $I$ intersects at least $k$ intervals from the set ${\cal I}_{s^i_q}\setminus\{I\}$, contradicting the assumption on $I$. The case $I'\in {\cal I}_2$ is symmetric to the case $I'\in {\cal I}_1$. This establishes property (2) and with it the claim that $s^i_q$ is a big node of $D^t_{k}$.
The path associated to component $C_i$ is defined as $P^i=s^{i}_{1}, \ldots, s^{i}_{p-2k+1}$.

Let $P$ denote the sequence of nodes of $D^t_{k}$ obtained by combining the paths $P^i$ into a single sequence of nodes preceded by $I_0$ and followed by $I_{n+1}$ in the natural order $I_0, P^1, \ldots, P^r, I_{n+1}$. Since $U\neq \emptyset$, the paths $P^i$ are pairwise node-disjoint, and none of them contains $I_0$ or $I_{n+1}$, path $P$ has at least $3$ nodes. Moreover, note that for each node $s$ of $P$ other than $I_0$ and $I_{n+1}$, the vertices of $G$ corresponding to intervals associated with $s$ all belong to the same component of $G[U]$, call it $C_s$.

We claim that $P$ is a path in $D^t_{k}$, that is, that every two consecutive nodes of $P$ form an arc in $D^t_{k}$.
Consider a pair $s,s'$ of consecutive nodes of $P$. Clearly, $s\neq I_{n+1}$. We consider three subcases depending on whether $s=I_0$, $s \in S$, or $s \in B$.

Suppose first that $s=I_0$. We claim that $(s,s')\in E_0$. Property (1) of the definition of the edges in $E_0$ clearly holds, as does (vacuously) property (3) (since $I_0 \notin B$). To show property (2), consider an interval $I\in {\cal I}$ such that $\max(s)<I<\min(s')$. Since $U$ is a total $k$-dominating set, interval $I$ intersects at least $k$ intervals from ${\cal I}^U$. Since $I <\min(s')$ and $\min(s')$ is the leftmost interval corresponding to a vertex of $U$, it follows that $I$ intersects the $k$ leftmost intervals from ${\cal I}^U$. All these intervals belong to vertices from component $C_1$, and therefore to ${\cal I}_{s'}$. This establishes property (2).
A similar argument shows that if $s' \in B$, then in order to make sure that the vertex corresponding to $\min(s')$ has at least $k$ neighbors in $U$, interval $\min(s')$ must intersect the $k$ intervals associated with $s'$ immediately following $\min(s')$ in the sequence.
This implies that the leftmost $k+1$ intervals associated with $s'$ pairwise intersect, thus establishing property (4). It follows that $(s,s') \in E_0$, as claimed.

Suppose now that $s \in S$.
We claim that $(s,s')\in E_0$. Property (1) of the definition of the edges in $E_0$ follows from the construction of $P$ and the fact that $C_s \neq C_{s'}$. Property (2) follows from the construction of $P$ together with the fact that $U$ is a total $k$-dominating set in $G$. Property (3) is satisfied vacuously. A similar argument as in the case $s=I_0$ establishes property (4). It follows that $(s,s') \in E_0$, as claimed.

Suppose now that $s \in B$. If $s'\in S$, then we conclude that $(s,s') \in E_0$ by symmetry with the case $s\in S$, $s'\in B$.
If $s'=I_{n+1}$, then we conclude that $(s,s') \in E_0$ by symmetry with the case $s=I_0$.
Let now $s' \in B$. If $C_s\neq C_{s'}$, then we can use similar arguments as in the case $s\in S$ to show that $(s,s')\in E_0$.
If $C_s=C_{s'}$, let $i\in \{1,\ldots, r\}$ be the index such that $C_s=C_i$. The construction of $P$ implies that $s$ and $s'$ are nodes of $P^i$ such that $s=s^i_q$ and $s'=s^i_{q+1}$ for some $q \in \{1, \ldots, p-2k\}$ where $p=|V(C_i)|$. The definitions of $s^i_q$ and $s^i_{q+1}$ now imply that $(s,s')\in E_1$, showing in particular that $(s,s')$ is an arc of $D^t_{k}$.

This shows that $P$ is a directed path from $I_0$ to $I_{n+1}$ in $D^t_{k}$, as claimed.
Furthermore, the definition of the length function $\ell$ and the
construction of $P$ imply that the length of $P$ equals the size of $U$, which is $c$.

\medskip
Now we establish the backward implication. Suppose that $D^t_{k}$ has a directed path $P$ from $I_0$ to $I_{n+1}$ of length $c$. Let $U$ be the set of all vertices $u\in V(G)$ such that the interval corresponding to $u$ is associated with some node of $P$. We claim that $U$ is a total $k$-dominating set in $G$ of size $c$.

Note that since neither of $I_0$ and $I_{n+1}$ is a big node, $(I_0,I_{n+1})\not\in E_1$; moreover, $(I_0,I_{n+1})\not\in E_0$ since condition $(2)$ in the definition
of an arc in $E_0$ fails. Therefore, $(I_0,I_{n+1})\not\in E(D^t_{k})$ and $P$ has at least $3$ nodes.
The set of nodes of $P$ can be uniquely partitioned into consecutive sets of nodes, say $W_0,W_1,\ldots, W_r, W_{r+1}$, such that
each $W_i$ is the node set of a maximal subpath $P'$ of $P$ such that $E(P')\subseteq E_1$.
(Equivalently, the $W_i$'s are the vertex sets of the components of the undirected graph underlying the digraph $P-E_0$.)
Note that $W_0 = \{I_0\}$ and $W_{r+1} = \{I_{n+1}\}$.
For all $i\in \{0,1,\ldots, r+1\}$,
let $U_i$ be the set of vertices of $G$ corresponding to
intervals associated with nodes in $W_i$.

We claim that for every $0\le i<j\le r+1$ and for every pair of intervals
$I\in {\cal I}^{U_i}$ and $J\in {\cal I}^{U_j}$ we have $I<J$ and $I\cap J = \emptyset$.
This can be proved by induction on $d = j-i$. If $d = 1$, then let
$(s,s')$ be the unique arc of $P$ connecting a vertex in $W_i$ with a vertex in $W_{i+1}$.
By the definition of the $W_i$'s, we have $(s,s')\not\in E_1$ and therefore $(s,s')\in E_0$. This
implies that every interval associated with $s$ is
smaller (with respect to ordering $<$) and disjoint from every interval associated with $s'$.
Consequently, the definitions of $W_i$, $W_{i+1}$, $U_i$, $U_{i+1}$, and the properties of the arcs in $E_0$ imply that
every interval in ${\cal I}^{U_i}$ is smaller (with respect to ordering $<$) and disjoint from every interval
in ${\cal I}^{U_{i+1}}$. The inductive step follows from the transitivity of the relation
on ${\cal I}'$ in which interval $I$ is in relation with interval $J$ if and only if
$I<J$ and $I\cap J = \emptyset$.

The above claim implies that no edge of $G$ connects a vertex in $U_i$ with a vertex in $U_j$ whenever $1\le i<j\le r$.
More specifically, we claim that for every $i\in \{1,\ldots, r\}$,
the subgraph of $G$ induced by $U_i$ is a component of $G[U]$.
This can be proved using the properties of the arcs in $D^t_{k}$, as follows.
Let $i\in \{1,\ldots, r\}$. Suppose first that $W_i\cap S \neq\emptyset$. Since every arc in $E_1$ connects a pair of big nodes,
we infer that $W_i = \{s\}$ for some $s\in S$. Using property (2) in the definition of a small node, we infer that $G[U_i]$ is connected.
Therefore, since no edge of $G$ connects a vertex in $U_i$ with a vertex in $U_j$ for $j\neq i$,
we infer that $G[U_i]$ is a component of $G[U]$, as claimed.
Suppose now that $W_i\cap S = \emptyset$, that is, $W_i\subseteq B$. Let $P'$ be the subpath of $P$ such that $V(P') = W_i$. Since $P'$ consists only of big nodes and only of arcs in $E_1$, we can use property (1) in the definition of a big node to infer that $G[U_i]$ is connected.
It follows that $G[U_i]$ is a component of $G[U]$ also in this case.

Let us now show that the size of $U$ equals the length of $P$. This will imply that $|U|\le c$.
For every $i\in \{1,\ldots, r+1\}$, let $P^i$ be the subpath of $P$ of consisting of all the arcs of $P$
entering a node in $W_i$. By construction, the paths $P^1,\ldots, P^{r+1}$ are pairwise arc-disjoint
and their union is $P$. For every $i\in \{1,\ldots, r+1\}$, let $\ell_i$ denote the the length of $P^i$.
The definitions of the $W_i$'s and of the length function imply that $\ell_i = \ell(s_i,s_i')+|W_i|-1$, where $(s_i,s_i')$ is the (unique) arc of $P$ such that
$s_i\not\in W_i$ and $s_i'\in W_i$.
Since $(s_i,s_i')\in E_0$, we have
$$\ell(s_i,s_i') = \left\{
                    \begin{array}{ll}
                     |{\cal I}_{s_i'}|, & \hbox{if $i\in \{1,\ldots, r\}$;} \\
                      0, & \hbox{if $i = r+1$.}
                    \end{array}
                  \right.$$
It follows that $$\ell_i = \left\{
                                       \begin{array}{ll}
                                         |{\cal I}_{s_i'}|+|W_i|-1, & \hbox{if $i\in \{1,\ldots, r\}$;} \\
                                         0, & \hbox{if $i = r+1$.}
                                       \end{array}
                                     \right.$$
Furthermore, we have $|{\cal I}_{s_i'}|+|W_i|-1 = |U_i|$ for all $i\in \{1,\ldots, r\}$, which implies
$\ell_i = |U_i|$. Since the subgraph of $G$ induced by $U_i$ is a component of $G[U]$,
we have $|U| = \sum_{i = 1}^r|U_i| = \sum_{i = 1}^{r+1}\ell_i$, which is exactly the length of $P$ (this follows from the fact that
the paths $P^1,\ldots, P^{r+1}$ are pairwise arc-disjoint and their union is $P$). Therefore, $|U| = c$.

It remains to show that $U$ is a total $k$-dominating set of $G$, that is, that every vertex
$u\in V(G)$ has at least $k$ neighbors in $U$.
Let $u\in V(G)$, let $I\in {\cal I}$ be the interval corresponding to $u$.
We need to show that $I$ intersects at least $k$ intervals from the set
${\cal I}^U\setminus\{I\}$.
We consider two cases depending on whether $u\in U$ or not.

\medskip
\noindent \emph{Case 1. $u\in U$.}
In this case, $I\in{\cal I}^U$ and $I$ is associated with some node of $P$.
Let $s\in V(P)$ be such a node. Note that $s\not\in \{I_0,I_{n+1}\}$.
By construction of $U$, we have ${\cal I}_s \subseteq {\cal I}^U$.
Suppose first that $s$ is a small node. Then, condition $(3)$ in the definition of a small node implies that
$I$ intersects at least $k$ intervals from the set ${\cal I}_s\setminus\{I\}$,
which is a subset of ${\cal I}^U\setminus\{I\}$.

Suppose now that $s$ is a big node. There exists a unique $i\in \{1,\ldots, r\}$ such that $s\in W_i$.
Note that $W_i$ is the node set of a subpath of $P$, say $P'$, consisting only of big nodes and arcs in $E_1$.
Let $I_{j_1},\ldots, I_{j_p}$ be the intervals corresponding to the vertices in $U_i$, ordered
increasingly. The fact that all arcs of $P'$ are in $E_1$ imply that
$V(P') = W_i = \{s_q\mid 1\le q\le p-2k+1\}$ where $s_q = (I_{j_{q}}, \ldots, I_{j_{q+2k-1}})$
and $E(P') = \{(s_q,s_{q+1})\mid 1\le q\le p-2k\}$.
Moreover, there exists a unique index $q\in \{1,\ldots, p\}$ such that
$I = I_{j_q}$.

Suppose first that $q\le k+1$. Let $(s',s'')$ be the (unique) arc of $P$ such that $s'\not\in W_i$ and $s''\in W_i$.
Since $(s',s'')\in E_0$ and $s''\in B$, condition $(4)$ in the definition of an arc in $E_0$
guarantees that the intervals $I_{j_1}, \ldots, I_{j_{k+1}}$ pairwise intersect.
Clearly, this implies that interval $I$ intersects at least $k$ intervals from the set ${\cal I}^U\setminus\{I\}$.

The case when $p-q\le k$ is symmetric to the case $q\le k+1$ and
can be analyzed using condition $(3)$ in the definition of an arc in $E_0$.

\begin{sloppypar}
Suppose now that $k+1 < q <p-k$. Let $\alpha = q-k$ and let $s_\alpha = (I_{j_{\alpha}}, \ldots, I_{j_{\alpha+2k-1}})$.
Then $s_\alpha\in V(P')$ and $s_\alpha$ is a big node of $D^t_{k}$. Moreover, $I_{j_{\alpha+k-1}} = I_{j_{q-1}} < I_{j_{q}} = I = I_{j_{\alpha+k}}$ and therefore condition $(2)$ in the definition of a big node implies that
$I$ intersects at least $k$ intervals from the set ${\cal I}_{s_\alpha} \setminus \{I\}$,
which is a subset of ${\cal I}^{U}\setminus \{I\}$.
\end{sloppypar}

\medskip
\noindent \emph{Case 2. $u\not\in U$.}
In this case, $I\not\in{\cal I}^U$.
Let $I^-$ denote the rightmost interval in the set $\{I'\mid  I'<I, I'\in {\cal I}^U\}$ if such an interval exists, otherwise let $I^- = I_0$.
Similarly, let $I^+$ denote the leftmost interval in the set $\{I'\mid I<I', I'\in {\cal I}^U\}$ if such an interval exists, otherwise let $I^+ = I_{n+1}$.
Note that $I^-<I<I^+$.  We consider several subcases depending on $I^-$ and $I^+$.

\medskip
\noindent \emph{Case 2.1. $I^- = I_0$}. Let $s$ be the successor of $I_0$ on $P$.
Then, $(I_0,s)\in E_0$ and since $I_0 = I^- < I < I^+ = \min(s)$, condition $(2)$ in the definition of an arc in $E_0$
implies that $I$ intersects at least $k$ intervals from ${\cal I}_x\cup {\cal I}_s$.
Since $I$ does not intersect any interval associated with $I_0$, we infer that
$I$ intersects at least $k$ intervals from ${\cal I}_s$, which is a subset of ${\cal I}^{U}\setminus \{I\}$.

\medskip
\noindent \emph{Case 2.2. $I^+ = I_{n+1}$.} This case is symmetric to Case 2.1.

\medskip
\noindent \emph{Case 2.3. $I_0<I^-<I^+ < I_{n+1}$ and $I^-\cap I^+ = \emptyset$}.
In this case, there exists a unique $i\in \{1,\ldots, r-1\}$
such that $I^-$ is associated with a node in $W_i$
and $I^+$ is associated with a node in $W_{i+1}$.
Let $(s',s'')$ be the (unique) arc of $P$ such that $s'\in W_i$ and $s''\in W_{i+1}$.
Then $(s',s'')\in E_0$ and since $\max(s') = I^- < I < I^+ = \min(s'')$, condition $(2)$ in the definition of an arc in $E_0$
implies that $I$ intersects at least $k$ intervals from ${\cal I}_{s'}\cup {\cal I}_{s''}$,
which is a subset of ${\cal I}^{U}\setminus \{I\}$.

\medskip
\noindent \emph{Case 2.4. $I_0<I^-<I^+ < I_{n+1}$ and $I^-\cap I^+ \neq \emptyset$}.
In this case, there exists a unique $i\in \{1,\ldots, r\}$ such that each of $I^-$, $I^+$ is associated with a node in $W_i$.
Furthermore, since $I^-$ and $I^+$ are consecutive intervals in ${\cal I}^U$, there exists a node $s\in W_i$ such that
both $I^-$ and $I^+$ are associated with $s$. This is clearly true if $W_i$ consists of a single node.
If $W_i$ consists of more than one node, then it consists of big nodes only, and the fact that all edges of $P$ connecting
two nodes in $W_i$ are in $E_1$ implies that a node $s$ with the desired property can be obtained by
defining $s$ as the node in $W_i$ closest to $I_{n+1}$ (along $P$) such that interval $I^-$ is associated with $s$.

Clearly, $s\in S\cup B$. We consider two further subcases.

\medskip
\noindent \emph{Case 2.4.1. $s\in S$}.
In this case, we have $\min(s)\le I^-<I<I^+\le \max(s)$ and condition $(3)$ in the definition of a small node implies that
interval $I$ intersects at least $k$ intervals from the set ${\cal I}_s\setminus\{I\}$, which is a subset of ${\cal I}^U\setminus\{I\}$.

\medskip
\noindent \emph{Case 2.4.2. $s\in B$}.
In this case, $W_i$ is the node set of a subpath of $P$ consisting of big nodes only.
Let $I_{j_1},\ldots, I_{j_p}$ be the intervals corresponding to the vertices in $U_i$, ordered
increasingly. There exists a unique index $q\in \{1,\ldots, p-1\}$ such that
$I^- = I_{j_q}$ and $I^+ = I_{j_{q+1}}$.

Suppose first that $q<k$. Let $(s',s'')$ be the (unique) arc of $P$ such that $s'\not\in W_i$ and $s''\in W_i$.
Then $(s',s'')\in E_0$ and $s''\in B$, therefore condition $(4)$ in the definition of an arc in $E_0$
guarantees that the intervals $I_{j_1}, \ldots, I_{j_{k+1}}$ pairwise intersect.
This condition implies that interval $I$ intersects each of the intervals $I_{j_1}, \ldots, I_{j_{k+1}}$, and therefore
also at least $k+1$ intervals from the set ${\cal I}^U\setminus\{I\}$.

The case when $p-q<k$ is symmetric to the case $q<k$ and can be analyzed using condition $(3)$ in the definition of an arc in $E_0$.

\begin{sloppypar}
Suppose now that $k\le q\le p-k$. Let $\alpha = q-k+1$ and let $s_\alpha = (I_{j_{\alpha}}, I_{j_{\alpha+1}},\ldots, I_{j_{\alpha+2k-1}})$.
A similar argument as in Case 1 shows that $s_\alpha\in W_i$ and $s$ is a big node of $D^t_{k}$.
Moreover, $I_{j_{\alpha+k-1}} = I_{j_q} = I^-<I<I^+ = I_{j_{q+1}} = I_{j_{\alpha+k}}$. Therefore,
condition $(2)$ in the definition of a big node implies that
$I$ intersects at least $k$ intervals from the set ${\cal I}_{s_\alpha}$, which is a subset of ${\cal I}^{U}\setminus \{I\}$.
\end{sloppypar}

This shows that $U$ is a total $k$-dominating set of $G$ and completes the proof.\hfill\qed

%
\section{Improving the Running Time}\label{sec:time}


We assume all notations from Section~\ref{sec:tkd-construction}. In particular, $G$ is a given $n$-vertex proper interval graph equipped with a proper interval model ${\cal I}$ and $(D^t_{k},\ell)$ is the derived edge-weighted directed acyclic graph with $\mathcal{O}(n^{2k})$ nodes.
We apply Proposition~\ref{prop:correctness} and show that a shortest $I_0, I_{n+1}$-path in $D^t_{k}$ can be computed in time $\mathcal{O}(n^{3k})$.
The main idea of the speedup relies on the fact that the algorithm avoids examining all arcs of the digraph. This is achieved by employing a dynamic programming approach based on a partition of a subset of the node set into $\mathcal{O}(n^{k})$ parts depending on the nodes' suffixes of length $k$.
The partition will enable us to efficiently compute minimum lengths of four types of directed paths in $D^t_k$, all starting in $I_0$ and ending in a specified vertex, vertex set, arc, or arc set.
In particular, a shortest $I_0, I_{n+1}$-path in $D^t_{k}$ will be also computed this way.

\begin{theorem}\label{thm:tkd-pig}
For every positive integer $k$, the total $k$-domination problem is solvable in time $\mathcal{O}(|V(G)|^{3k})$ in the class of proper interval graphs.
\end{theorem}

\begin{proof}
In order to describe the algorithm in detail, we need to introduce some notation. Given a node $s\in S\cup B$, say $s = (I_{i_1},\ldots, I_{i_q})$ (recall that $k+1\le q\le 2k$), we define its \emph{$k$-suffix} of $s$ as the sequence $(I_{i_{q-k+1}}, \ldots, I_{i_q})$ and denote it by $\suf_k(s)$.

The algorithm proceeds as follows. First, it computes the node set of $D^t_{k}$ and
a subset $B'$ of the set of big nodes consisting of precisely those nodes $s\in B$
satisfying condition $(3)$ in the definition of $E_0$ (that is,
the rightmost $k+1$ intervals associated with $s$ pairwise intersect).
Next, it computes a partition $\{A_\sigma\mid \sigma\in \Sigma\}$ of $S\cup B'$ defined by
$\Sigma = \{\suf_k(s): s\in S\cup B'\}$ and $A_\sigma =  \{s\in S\cup B'\mid \suf_k(s) = \sigma\}$ for all $\sigma\in \Sigma$.

The algorithm also computes the arc set $E_1$. On the other hand, the arc set $E_0$ is not generated explicitly, except for the arcs in $E_0$ with tail $I_0$ or head $I_{n+1}$. Using dynamic programming, the algorithm will compute the following values.
\begin{enumerate}[(i)]
  \item For all $s\in V(D^t_{k})\setminus\{I_0\}$, let $p^0_s$ denote the minimum $\ell$-length of a directed \hbox{$I_0$,$s$-path} in $D^t_k$ ending with an arc from $E_0$.
  \item For all $s\in V(D^t_{k})\setminus\{I_0\}$, let $p_s$ denote the minimum $\ell$-length of a directed \hbox{$I_0$,$s$-path} in $D^t_k$.
  \item For all $e\in E_1$, let $p_e$ denote the minimum $\ell$-length of a directed path in $D^t_k$ starting in $I_0$ and ending with $e$.
  \item For all $\sigma\in \Sigma$, let
$p_\sigma$ denote the minimum $\ell$-length of a directed path in $D^t_k$ starting in $I_0$ and ending in $A_\sigma$.
\end{enumerate}
In all cases, if no path of the corresponding type exists, we set the value of the respective $p^0_s$,  $p_s$, $p_e$, or $p_\sigma$ to $\infty$.

Clearly, once all the $p^0_s$, $p_s$, $p_e$, and $p_\sigma$ values will be  computed, the length of a shortest $I_0, I_{n+1}$-path in $D^t_{k}$ will e given by $p_{I_{n+1}}$.

The above values can be computed using the following recursive formulas:
\begin{enumerate}[(i)]
\item $p^0_s$ values:
\begin{itemize}
  \item For $s \in S\cup B$, let $\Sigma_s = \{\sigma\in \Sigma\mid (\tilde{s},s)\in E_0\text{ for some }\tilde{s}\in A_\sigma\}$ %
and set $$p^0_s = \left\{
            \begin{array}{ll}
              |{\cal I}_s|, & \hbox{if $(I_0,s)\in E_0$;} \\
              \underset{\sigma\in \Sigma_s}{\min}p_\sigma + |{\cal I}_s|, & \hbox{if $(I_0,s)\not\in E_0$ and $\Sigma_s\neq \emptyset$;}\\
              \infty, & \hbox{otherwise.}
            \end{array}
          \right.$$
  \item For $s = I_{n+1}$, let $p_s^0 = \underset{(\tilde{s},s)\in E_0}{\min}p_{\tilde{s}}$.
\end{itemize}
  \item $p_s$ values: For all $s\in V(D^t_k)\setminus\{I_0\}$, we have $p_s  = \min\left\{p^0_s, \underset{(\tilde{s},s)\in E_1}{\min}p_{(\tilde{s},s)}\right\}$.

\smallskip
  \item $p_e$ values: For all $e = (s,s')\in E_1$, we have $p_e = p_s+1$.

\smallskip
   \item $p_\sigma$ values: For all $\sigma\in \Sigma$, we have
$p_\sigma  = \underset{s\in A_\sigma}{\min}p_s$.
\end{enumerate}
The above formulas can be computed following any topological sort of $D^t_k$ such that
if $s,s'\in S\cup B$ are such that $\suf_k(s)\neq \suf_k(s')$ and $\suf_k(s)$ is lexicographically smaller than $\suf_k(s')$, then $s$ appears strictly before $s'$ in the ordering.
When the algorithm processes a node $s\in V(D^t_k)\setminus\{I_0\}$, it computes the values of $p_s^0$, $p_e$ for all $e = (\tilde{s},s)\in E_1$, and $p_s$, in this order.
For every $\sigma\in \Sigma$, the value of $p_\sigma$ is computed as soon as
the values of $p_s$ are known for all $s\in A_\sigma$.

\medskip
{\it Correctness of the algorithm.} We will justify the recursive formula for the $p^0_s$ values when $s\in S\cup B$; all other recursive formulas follow directly from the definitions of the values involved
and length function $\ell$ (cf.~equation~\eqref{eq:length} on page~\pageref{eq:length}).
Let $s\in S\cup B$ and consider a directed $I_0$,$s$-path $P$ in $D^t_k$ ending with an arc $(\tilde s,s)$ from $E_0$. Note that by the definition of length function $\ell$, we have $\ell(\tilde s,s) = |{\cal I}_s|$, independently of $\tilde s$. Thus, $p^0_s\ge |{\cal I}_s|$, with equality if and only if $(I_0,s)\in E_0$. Suppose now that $(I_0,s)\not\in E_0$. Then $\tilde s\in S\cup B'$ and setting $\sigma = \suf_k(\tilde{s})$, we have $\sigma\in \Sigma$ and $\tilde s\in A_\sigma$, which implies $\sigma\in \Sigma_s$.

We show next that for every $\bar{s}\in A_{\sigma}$ we have $(\bar{s}, s)\in E_0$. Let $\bar{s}\in A_{\sigma}$. Then,  $\bar{s}\in S\cup B'$ and
$\suf_k(\bar{s}) = \sigma = \suf_k(\tilde{s})$, which implies that
$\max(\bar{s}) = \max(\tilde{s})$. Let us now verify, using the fact that $(\tilde s,s)\in E_0$ and the corresponding conditions (1)--(4) in the definition of $E_0$, that conditions (1)--(4) in the definition of $E_0$ also hold when applied to the pair $(
\bar{s}, s)$. Condition (1) follows from the fact that $\max(\bar{s}) = \max(\tilde{s})$.
For condition (2), let $I\in {\cal I}$ be an interval such that $\max(\bar{s})<I<\min(s)$.
Then $\max(\tilde{s})<I<\min(s)$ and thus $I$ intersects at least $k$ intervals from
${\cal I}_{\tilde{s}}\cup {\cal I}_{s}$. This implies that $I$ intersects at least $k$ intervals from
the set ${\cal I}_{\suf_k(\tilde{s})}\cup {\cal I}_{s}$, where ${\cal I}_{\suf_k(\tilde{s})}$ denotes the set of rightmost $k$ intervals associated with $\tilde{s}$. Since $\suf_k(\bar{s}) = \suf_k(\tilde{s})$, the set ${\cal I}_{\suf_k(\tilde{s})}$ coincides with the set of
rightmost $k$ intervals associated with $\bar{s}$, and therefore $I$ intersects at least $k$ intervals from
the set ${\cal I}_{\bar{s}}\cup {\cal I}_{s}$, yielding condition (2).
Condition (3) holds since if $\bar{s}\in B$ then $\bar{s}\in B'$. Finally, condition (4) holds since it holds for the arc
$(\tilde s,s)$ and depends only on $s$. Thus, conditions (1)--(4) in the definition of $E_0$ hold when applied to the pair $(\bar{s}, s)$, which implies $(\bar{s}, s)\in E_0$, as claimed. Altogether, this justifies that $p^0_s = \underset{\tilde{\sigma}\in \Sigma_s}{\min}p_{\tilde{\sigma}} + |{\cal I}_s|$ if
$(I_0,s)\not\in E_0$ and $\Sigma_s\neq\emptyset$.

\smallskip
{\it Analysis of the running time.}
The node set of $D^t_k$, including sets $S$, $B$, and $B'$, as well as the set of arcs in $E_0$ with tail $I_0$ or head $I_{n+1}$ can be computed in time $\mathcal{O}(n^{2k+1})$. In the same time the partition $S\cup B' = \{A_\sigma\mid \sigma\in \Sigma\}$ can be computed, along with a topological sort of $D^t_k$ as specified above (for example by iterating over the nodes in $S\cup B$ and building a radix tree with respect to their $k$-suffixes).

When processing a node $s\in S\cup B$, the set $\Sigma_s$ can be computed in time
$\mathcal{O}(n^{k})$ by verifying for each $\sigma\in \Sigma$, whether an arbitrarily chosen node $\tilde{s}\in A_\sigma$ satisfies $(\tilde{s},s)\in E_0$. (As noted above, this property
is independent of the choice of $\tilde{s}$.)
In the same time $\mathcal{O}(n^{k})$, the value of $p_s^0$ can be computed.
The values of $p_e$ for all $e = (\tilde{s},s)\in E_1$ as well as $p_s$ can be computed in time
proportional to $d_{E_1}^-(s)+1$, where $d_{E_1}^-(s)$ is the in-degree of $s$
in the spanning subdigraph of $D^t_{k}$ with arc set $E_1$.
Processing of node $I_{n+1}$ can be done in time proportional to its in-degree in
$D^t_k$, which is in $\mathcal{O}(n^{2k})$.
We conclude that all the nodes of $D^t_k$ can be processed in time
\hbox{$\mathcal{O}\left(\sum_{s\in S\cup B}\left(n^{k}+d_{E_1}^-(s)+1\right)\right)+\mathcal{O}\left(n^{2k}\right) =
\mathcal{O}(n^{3k})$}, since, in particular, $d_{E_1}^-(s)\le n$ for all $s\in S\cup B$.
The values of $p_\sigma$ for all $\sigma\in \Sigma$ can be computed in overall time
$\mathcal{O}(\sum_{\sigma\in \Sigma}|A_\sigma|) = \mathcal{O}(n^{2k})$.
Thus, the overall time complexity of the algorithm is
$\mathcal{O}(n^{3k})$, as claimed.
\end{proof}

\section{Modifying the Approach for $k$-Domination}\label{sec:kd}

With minor modifications of the definitions of small nodes, big nodes, and arcs in $E_0$ of the derived digraph, the approach developed in Sections~\ref{sec:tkd-construction}--\ref{sec:time} for total $k$-domination leads to an analogous result for $k$-domination.

Given a proper interval graph $G$ equipped with a proper interval model ${\cal I}$ (as in Section~\ref{sec:tkd-construction}), we construct an edge-weighted directed acyclic graph denoted by $D_k$. The digraph $D_k$ is defined the same way as $D^t_{k}$ (see Section~\ref{sec:tkd-construction}), except for the following:
\begin{itemize}
\item The set $S$ of small nodes now consists exactly of those increasing sequences $s = (I_{i_1},\ldots, I_{i_q})$ of intervals from ${\cal I}$ such that:
\begin{enumerate}[(1)]
  {\setlength\itemindent{0.2cm}\item $1\le q \le 2k-1$,
  }{\setlength\itemindent{0.2cm}\item $I_{i_j}\cap I_{i_{j+1}}\neq\emptyset$ for all $j\in \{1,\ldots, q-1\}$, and
  }{\setlength\itemindent{0.2cm}\item every interval $I\in {\cal I}\setminus {\cal I}_s$  such that $\min(s)< I < \max(s)$ intersects at least $k$ intervals from the set ${\cal I}_s$.}
\end{enumerate}

\smallskip
\item The set $B$ of big nodes consists exactly of those
increasing sequences $s = (I_{i_1},\ldots, I_{i_{2k}})$ of intervals from ${\cal I}$ of length $2k$ such that:
\begin{enumerate}[(1)]
  {\setlength\itemindent{0.2cm}\item $I_{i_j}\cap I_{i_{j+1}}\neq\emptyset$ for all $j\in \{1,\ldots, 2k-1\}$ and
  }{\setlength\itemindent{0.2cm}\item every interval $I\in {\cal I}\setminus {\cal I}_s$  such that $I_{i_k}< I < I_{i_{k+1}}$ intersects at least $k$ intervals from the set ${\cal I}_s$.}
\end{enumerate}

\smallskip
\item Arc set $E_0$ consists exactly of those ordered pairs $(s,s')\in V(D_k)\times V(D_k)$ such that:
\begin{enumerate}[(1)]
  {\setlength\itemindent{0.2cm}\item $\max(s)< \min(s')$ and $\max(s)\cap \min(s') = \emptyset$,
  }{\setlength\itemindent{0.2cm}\item every interval $I\in {\cal I}$ such that $\max(s)<I<\min(s')$ intersects at least $k$ intervals from
  ${\cal I}_s\cup {\cal I}_{s'}$,
  }{\setlength\itemindent{0.2cm}\item if $s=(I_{j_1},\ldots,I_{j_{2k}})$ is a big node, then every interval $I\in {\cal I}\setminus {\cal I}_s$  such that $I_{j_{k+1}}< I < I_{j_{2k}}$ intersects at least $k$ intervals in ${\cal I}_s$, and
  }{\setlength\itemindent{0.2cm}\item if $s'=(I_{j_1},\ldots,I_{j_{2k}})$ is a big node, then every interval $I\in {\cal I}\setminus {\cal I}_{s'}$  such that $I_{j_{1}}< I < I_{j_{k}}$ intersects at least $k$ intervals in ${\cal I}_{s'}$.}
\end{enumerate}
\end{itemize}

The arc set $E_1$ is defined analogously as in the case of $D^t_k$ (using, of course, ordered pairs from $V(D_k)$) and the length function $\ell(s,s')$ on the arcs $(s,s')$ of $D_k$ is defined using~\eqref{eq:length} (see Section~\ref{sec:tkd-construction}).

A similar approach as that used to prove Proposition~\ref{prop:correctness} yields the following.

\begin{proposition}\label{prop:correctness-k-domination}
Given a proper interval graph $G$ and a positive integer $k$, let $D_k$ be the directed graph constructed as above.
Then $G$ has a $k$-dominating set of size $c$ if and only if $D_k$ has a directed path from $I_0$ to $I_{n+1}$ of length $c$.
\end{proposition}

\begin{proof}
We assume notation from Sections~\ref{sec:tkd-construction} and~\ref{sec:kd}, and, additionally, fix the following notation useful for both directions of the proof: for $X\subseteq V(G)$, we denote by ${\cal I}^X$ the set of intervals in ${\cal I}$ corresponding to vertices in $X$.

\medskip
First we establish the forward implication. Suppose that $G$ has a $k$-dominating set $U$ of size $c$. The components of $G[U]$ can be naturally ordered from left to right, say as $C_1,\ldots, C_r$. To each component $C_i$ we will associate a path $P^i$ in $D_k$ defined as a sequence of nodes. The desired directed path $P$ from $I_0$ to $I_{n+1}$ in $D_k$ will be then obtained by combining the paths $P^i$ into a single sequence of nodes preceded by $I_0$ and followed by $I_{n+1}$.

We say that a component $C_i$ is \emph{small} if $|V(C_i)|<2k$ and
\emph{big}, otherwise. To every small component $C_i$ we associate a sequence $s^i=(I_{i_{1}}, \ldots, I_{i_{p}})$, consisting of the $p=|V(C_i)|$ intervals corresponding to the vertices of $C_i$, ordered increasingly. We claim that $s^i$ is a small node of $D_k$. Property (1) of the definition of a small node is satisfied by definition. Property (2) follows from the fact that ${\cal I}$ is a proper interval model of $G$ and $C_i$ is connected. To show that $s^i$ satisfies property (3) of the definition of a small node, consider an interval $I \in {\cal I}\setminus {\cal I}_s$ such that $\min(s^i) < I < \max(s^i)$ and let $u$ be the vertex of $G$ corresponding to $I$. Note that since $I\not\in {\cal I}_s$, we have $I\not\in {\cal I}^U$ and hence $u\not\in U$. In particular, since $U$ is a $k$-dominating set in $G$, vertex $u$ has at least $k$ neighbors in $U$. Since $\min(s^i) < I < \max(s^i)$, all the neighbors of $u$ in $U$ must belong to $C_i$, more specifically, $N(u) \cap U \subseteq V(C_i)$. It follows that $I$ intersects at least $k$ intervals from the set ${\cal I}_{s^i}$, establishing also property (3) and with it the claim that $s^i$ is a small node of $D_k$.
The path $P^i$ associated to component $C_i$ is the one-node path having $s^i$ as a node.

Let now $C_i$ be a big component and let $p = |V(C_i)|$. Then $p\ge 2k$.
Let $I_{j_1},\ldots,  I_{j_p}$ be the intervals corresponding to the vertices of $C_i$, ordered increasingly.
For every $q \in \{1,\ldots, p-2k+1\}$, let $s^i_q$ denote the subsequence of these intervals of length $2k$ starting at $q$-th interval,
that is, $s^i_q = (I_{j_q}, I_{j_{q+1}},\ldots, I_{j_{q+2k-1}})$. We claim that for each $q\in \{1,\ldots, p-2k+1\}$,  $s^{i}_{q}$ is a big node of $D_k$.
Property (1) follows from the fact that ${\cal I}$ is a proper interval model of $G$ and $C_i$ is connected.  To show that $s^i_q$ satisfies property (2) of the definition of a big node, consider an interval $I \in {\cal I}\setminus {\cal I}_s$ such that $I_{j_{q+k-1}} < I < I_{j_{q+k}}$ and let $u$ be the vertex of $G-U$ corresponding to $I$. Since $U$ is a $k$-dominating set, vertex $u$ has at least $k$ neighbors in $U$. Since $I_{j_{1}} < I < I_{j_{p}}$, all the neighbors of $u$ in $U$ must belong to $C_i$. It follows that $I$ intersects at least $k$ intervals from the set $\{I_{j_{1}}, \ldots, I_{j_{p}}\}\setminus\{I\}$. Suppose for a contradiction that $I$ intersects strictly less that $k$ intervals from the set ${\cal I}_{s^i_q}\setminus\{I\}=\{I_{j_{q}}, I_{j_{q+1}}, \ldots, I_{j_{q+2k-1}}\}\setminus\{I\}$. Since $I$ intersects at least $k$ intervals from the set $\{I_{j_{1}}, \ldots, I_{j_{p}}\}\setminus\{I\}$, there is an interval, call it $I'$, in the set $\{I_{j_{1}}, \ldots, I_{j_{p}}\} \setminus {\cal I}_{s^i_q}$ such that $I \cap I' \neq \emptyset$. Note that
$\{I_{j_{1}}, \ldots, I_{j_{p}}\} \setminus {\cal I}_{s^i_q} = {\cal I}_1\cup {\cal I}_2$ where
${\cal I}_1 = \{I_{j_1}, \ldots, I_{j_{q-1}}\}$ and ${\cal I}_2 = \{I_{j_{q+2k}}, \ldots, I_{j_{p}}\}$. Suppose first that $I'\in {\cal I}_1$, that is, $I' = I_{j_{\alpha}}$ for some $\alpha\in \{1,\ldots, q-1\}$.
Since ${\cal I}$ is a proper interval model of $G$, conditions $I_{j_{q+k-1}}< I< I_{j_{q+k}}$ and $I_{j_{q+k-1}}\cap I_{j_{q+k}}\neq \emptyset$ imply that interval $I$ has non-empty intersection with interval $I_{j_{q+k-1}}$. Similarly, the fact that $I$ intersects both $I'=I_{j_{\alpha}}$ and $I_{j_{q+k-1}}$ implies that $I$ also intersects each of the intervals in the set $\{I_{j_{\alpha}}, I_{j_{\alpha+1}},\ldots, I_{j_{q+k-1}}\}$; in particular, $I$ intersects each of the $k$ intervals in the set $\{I_{j_{q}}, I_{j_{q+1}},\ldots, I_{j_{q+k-1}}\}$, which is a subset of ${\cal I}_{s^i_q}$. This  contradicts the assumption on $I$. The case $I'\in {\cal I}_2$ is symmetric to the case $I'\in {\cal I}_1$.
This establishes property (2) and with it the claim that $s^i_q$ is a big node of $D_k$.
The path associated to component $C_i$ is defined as $P^i=s^{i}_{1}, \ldots, s^{i}_{p-2k+1}$.

Let $P$ denote the sequence of nodes of $D_k$ obtained by combining the paths $P^i$ into a single sequence of nodes preceded by $I_0$ and followed by $I_{n+1}$, in the natural order $I_0, P^1, \ldots, P^r, I_{n+1}$. Since $U\neq \emptyset$, the paths $P^i$ are pairwise node-disjoint, and none of them contains $I_0$ or $I_{n+1}$,
path $P$ has at least $3$ nodes. Moreover, note that for each node $s$ of $P$ other than $I_0$ and $I_{n+1}$, the vertices of $G$ corresponding to intervals associated with $s$ all belong to the same component of $G[U]$, call it $C_s$.

We claim that $P$ is a path in $D_k$, that is, that every two consecutive nodes of $P$ form an arc in $D_k$.
Consider a pair $s,s'$ of consecutive nodes of $P$. Clearly, $s\neq I_{n+1}$. We consider three subcases depending on whether $s=I_0$, $s \in S$, or $s \in B$.

Suppose first that $s=I_0$. We claim that $(s,s')\in E_0$. Property (1) of the definition of the edges in $E_0$ clearly holds, as does (vacuously) property (3) (since $I_0 \notin B$). To show property (2), consider an interval $I\in {\cal I}$ such that $\max(s)<I<\min(s')$. Since $U$ is a $k$-dominating set and $I$ corresponds to a vertex of $G$ not in $U$, interval $I$ intersects at least $k$ intervals from ${\cal I}^U$. Since $I <\min(s')$ and $\min(s')$ is the leftmost interval corresponding to a vertex of $U$, it follows that $I$ intersects the $k$ leftmost intervals from ${\cal I}^U$. All these intervals belong to vertices from component $C_1$, and therefore to ${\cal I}_{s'}$. This establishes property (2).
It remains to verify property $(4)$. Let $s' \in B$, say $s'=(I_{j_1},\ldots,I_{j_{2k}})$, and consider an interval
 $I\in {\cal I}\setminus {\cal I}_{s'}$  such that $I_{j_{1}}< I < I_{j_{k}}$. We claim that $I\in {\cal I}\setminus {\cal I}^U$.
Suppose that this is not the case. The construction of $P$ implies that ${\cal I}_{s'}$ consists of the $2k$ leftmost intervals corresponding to vertices in $U$. In particular, if $I\in {\cal I}^U$, then condition $I_{j_{1}}< I < I_{j_{k}}$ implies that
$I = I_{j\alpha}$ for some $1<\alpha<k$, hence $I\in {\cal I}_{s'}$, a contradiction. Since $I\in {\cal I}\setminus {\cal I}^U$ and $U$ is a $k$-dominating set in $G$, interval $I$ intersects at least $k$ intervals from ${\cal I}^U$. Again, since $I_{j_{1}},\ldots, I_{j_{2k}}$ are the $2k$ leftmost
intervals corresponding to vertices in ${\cal I}^U$ and $I < I_{j_{k}}$, the fact that $I$ intersects at least $k$ intervals from ${\cal I}^U$ implies that
$I$ also intersects at least $k$ intervals from the set $\{I_{j_1},\ldots, I_{j_{2k-1}}\}$, which is a subset of ${\cal I}_{s'}$. This establishes property (4). It follows that $(s,s') \in E_0$, as claimed.

Suppose now that $s \in S$.
We claim that $(s,s')\in E_0$. Property (1) of the definition of the edges in $E_0$ follows from the construction of $P$ and the fact that $C_s \neq C_{s'}$. Property (2) follows from the construction of $P$ together with the fact that $U$ is a $k$-dominating set in $G$. Property (3) is satisfied vacuously. A similar argument as in the case $s=I_0$ establishes property (4) (this time using the fact that if $s'\in B$, then ${\cal I}_{s'}$ consists of the $2k$ leftmost intervals in ${\cal I}^{C_i}$ where $C_i$ is the component of $G[U]$ containing the vertices corresponding to
intervals associated with $s'$).
It follows that $(s,s') \in E_0$, as claimed.

Suppose now that $s \in B$. If $s'\in S$, then we conclude that $(s,s') \in E_0$ by symmetry with the case $s\in S$, $s'\in B$.
If $s'=I_{n+1}$, then we conclude that $(s,s') \in E_0$ by symmetry with the case $s=I_0$.
Let now $s' \in B$. If $C_s\neq C_{s'}$, then we can use similar arguments as in the case $s\in S$ to show that $(s,s')\in E_0$.
If $C_s=C_{s'}$, let $i\in \{1,\ldots, r\}$ be the index such that $C_s=C_i$. The construction of $P$ implies that $s$ and $s'$ are nodes of $P^i$ such that $s=s^i_q$ and $s'=s^i_{q+1}$ for some $q \in \{1, \ldots, p-2k\}$ where $p=|V(C_i)|$. The definitions of $s^i_q$ and $s^i_{q+1}$ now imply that $(s,s')\in E_1$, showing in particular that $(s,s')$ is an arc of $D_k$.

This shows that $P$ is a directed path from $I_0$ to $I_{n+1}$ in $D_k$, as claimed. Furthermore, the construction of $P$ implies that the length of $P$ equals the size of $U$, which is $c$.

\medskip
Now we establish the converse implication. Suppose that $D_k$ has a directed path $P$ from $I_0$ to $I_{n+1}$ of length $c$.
Let $U$ be the set of all vertices $u\in V(G)$ such that the interval corresponding to $u$ is associated with some node of $P$.
We claim that $U$ is a $k$-dominating set in $G$ of size $c$.

Note that since neither of $I_0$ and $I_{n+1}$ is a big node, $(I_0,I_{n+1})\not\in E_1$; moreover, $(I_0,I_{n+1})\not\in E_0$ since condition $(2)$ in the definition
of an arc in $E_0$ fails. Therefore, $(I_0,I_{n+1})\not\in E(D_k)$ and $P$ has at least $3$ nodes.
The set of nodes of $P$ can be uniquely partitioned into consecutive sets of nodes, say $W_0,W_1,\ldots, W_r, W_{r+1}$, such that
each $W_i$ is the node set of a maximal subpath $P'$ of $P$ such that $E(P')\subseteq E_1$.
(Equivalently, the $W_i$'s are the vertex sets of the components of the undirected graph underlying the digraph $P-E_0$.)
Note that $W_0 = \{I_0\}$ and $W_{r+1} = \{I_{n+1}\}$.
For all $i\in \{0,1,\ldots, r+1\}$,
let $U_i$ be the set of vertices of $G$ corresponding to
intervals in ${\cal I}'$ associated with nodes in $W_i$.

We claim that for every $0\le i<j\le r+1$ and for every pair of intervals
$I\in {\cal I}^{U_i}$ and $J\in {\cal I}^{U_j}$ we have $I<J$ and $I\cap J = \emptyset$.
This can be proved by induction on $d = j-i$. If $d = 1$, then let
$(s,s')$ be the unique arc of $P$ connecting a vertex in $W_i$ with a vertex in $W_{i+1}$.
By the definition of the $W_i$'s, we have $(s,s')\not\in E_1$ and therefore $(s,s')\in E_0$, which implies that every interval associated with $s$ is
smaller (with respect to ordering $<$) and disjoint from every interval associated with $s'$.
Consequently, the definitions of $W_i$, $W_{i+1}$, $U_i$, $U_{i+1}$, and the properties of the arcs in $E_1$ imply that
every interval in ${\cal I}^{U_i}$ is smaller (with respect to ordering $<$) and disjoint from every interval
in ${\cal I}^{U_{i+1}}$. The inductive step follows from the transitivity of the relation
on ${\cal I}'$ in which interval $I$ is in relation with interval $J$ if and only if
$I<J$ and $I\cap J = \emptyset$.

The above claim implies that no edge of $G$ connects a vertex in $U_i$ with a vertex in $U_j$ whenever $1\le i<j\le r$.
More specifically, we claim that for every $i\in \{1,\ldots, r\}$,
the subgraph of $G$ induced by $U_i$ is a component of $G[U]$.
This can be proved using the properties of the arcs in $D_k$, as follows.
Let $i\in \{1,\ldots, r\}$. Suppose first that $W_i\cap S \neq\emptyset$. Since every arc in $E_1$ connects a pair of big nodes,
we infer that $W_i = \{s\}$ for some $s\in S$. Using property (2) in the definition of a small node, we infer that $G[U_i]$ is connected.
Therefore, since no edge of $G$ connects a vertex in $U_i$ with a vertex in $U_j$ for $j\neq i$,
we infer that $G[U_i]$ is a component of $G[U]$, as claimed.
Suppose now that $W_i\cap S = \emptyset$, that is, $W_i\subseteq B$.
Let $P'$ be the subpath of $P$ such that $V(P') = W_i$. Since $P'$ consists only of big nodes and
only of arcs in $E_1$, we can use property (1) in the definition of a big node to infer that $G[U_i]$ is connected.
It follows that $G[U_i]$ is a component of $G[U]$ also in this case.

Let us now show that the size of $U$ equals the length of $P$. This will imply that $|U|\le c$.
For every $i\in \{1,\ldots, r+1\}$, let $P^i$ be the subpath of $P$ of consisting of all the arcs of $P$
entering a node in $W_i$. By construction, the paths $P^1,\ldots, P^{r+1}$ are pairwise arc-disjoint
and their union is $P$. For every $i\in \{1,\ldots, r+1\}$, let $\ell_i$ denote the the length of $P^i$.
The definitions of the $W_i$'s and of the length function imply that $\ell_i = \ell(s_i,s_i')+|W_i|-1$, where $(s_i,s_i')$ is the (unique) arc of $P$ such that
$s_i\not\in W_i$ and $s_i'\in W_i$.
Since $(s_i,s_i')\in E_0$, we have
$$\ell(s_i,s_i') = \left\{
                    \begin{array}{ll}
                     |{\cal I}_{s_i'}|, & \hbox{if $i\in \{1,\ldots, r\}$;} \\
                      0, & \hbox{if $i = r+1$.}
                    \end{array}
                  \right.$$
It follows that $$\ell_i = \left\{
                                       \begin{array}{ll}
                                         |{\cal I}_{s_i'}|+|W_i|-1, & \hbox{if $i\in \{1,\ldots, r\}$;} \\
                                         0, & \hbox{if $i = r+1$.}
                                       \end{array}
                                     \right.$$
Furthermore, we have $|{\cal I}_{s_i'}|+|W_i|-1 = |U_i|$ for all $i\in \{1,\ldots, r\}$, which implies
$\ell_i = |U_i|$. Since the subgraph of $G$ induced by $U_i$ is a component of $G[U]$,
we have $|U| = \sum_{i = 1}^r|U_i| = \sum_{i = 1}^{r+1}\ell_i$, which is exactly the length of $P$ (this follows from the fact that
the paths $P^1,\ldots, P^{r+1}$ are pairwise arc-disjoint and their union is $P$). Therefore, $|U| = c$.

It remains to show that $U$ is a $k$-dominating set of $G$, that is, that every vertex
$u\in V(G)\setminus U$ has at least $k$ neighbors in $U$.
Let $u\in V(G)\setminus U$ and let $I\in {\cal I}$ be the interval corresponding to $u$.
We need to show that $I$ intersects at least $k$ intervals from the set
${\cal I}^U$. Note that $I\not\in {\cal I}^U$.

Let $I^-$ denote the rightmost interval in the set $\{I'\mid I'<I, I'\in {\cal I}^U\}$ if such an interval exists, otherwise let $I^- = I_0$.
Similarly, let $I^+$ denote the leftmost interval in the set $\{I'\mid  I<I', I'\in {\cal I}^U\}$ if such an interval exists, otherwise let $I^+ = I_{n+1}$.
Note that $I^-<I<I^+$.  We consider several subcases depending on $I^-$ and $I^+$.

\medskip
\noindent \emph{Case 1. $I^- = I_0$}. Let $s$ be the successor of $I_0$ on $P$.
Then, $(I_0,s)\in E_0$ and since $I_0 = I^- < I < I^+ = \min(s)$, condition $(2)$ in the definition of an arc in $E_0$
implies that $I$ intersects at least $k$ intervals from ${\cal I}_x\cup {\cal I}_s$.
Since $I$ does not intersect any interval associated with $I_0$, we infer that
$I$ intersects at least $k$ intervals from ${\cal I}_s$, which is a subset of ${\cal I}^U$.

\medskip
\noindent \emph{Case 2. $I^+ = I_{n+1}$.} This case is symmetric to Case 1.

\medskip
\noindent \emph{Case 3. $I_0<I^-<I^+ < I_{n+1}$ and $I^-\cap I^+ = \emptyset$}.
In this case, there exists a unique $i\in \{1,\ldots, r-1\}$
such that $I^-$ is associated with a node in $W_i$
and $I^+$ is associated with a node in $W_{i+1}$.
Let $(s',s'')$ be the (unique) arc of $P$ such that $s'\in W_i$ and $s''\in W_{i+1}$.
Then $(s',s'')\in E_0$ and since $\max(s') = I^- < I < I^+ = \min(s'')$, condition $(2)$ in the definition of an arc in $E_0$
implies that $I$ intersects at least $k$ intervals from ${\cal I}_{s'}\cup {\cal I}_{s''}$, which is a subset of ${\cal I}^U$.

\medskip
\noindent \emph{Case 4. $I_0<I^-<I^+ < I_{n+1}$ and $I^-\cap I^+ \neq \emptyset$}.
In this case, there exists a unique $i\in \{1,\ldots, r\}$ such that each of $I^-$, $I^+$ is associated with a node in $W_i$.
Furthermore, since $I^-$ and $I^+$ are consecutive intervals in ${\cal I}^U$, there exists a node $s\in W_i$ such that
both $I^-$ and $I^+$ are associated with $s$. This is clearly true if $W_i$ consists of a single node.
If $W_i$ consists of more than one node, then it consists of big nodes only, and the fact that all edges of $P$ connecting
two nodes in $W_i$ are in $E_1$ implies that a node $s$ with the desired property can be obtained by
defining $s$ as the node in $W_i$ closest to $I_{n+1}$ (along $P$) such that interval $I^-$ is associated with $s$.

Clearly, $s\in S\cup B$. We consider two further subcases.

\medskip
\noindent \emph{Case 4.1. $s\in S$}.
In this case, we have $\min(s)\le I^-<I<I^+\le \max(s)$ and condition $(3)$ in the definition of a small node implies that
interval $I$ intersects at least $k$ intervals from the set ${\cal I}_s$, which is a subset of ${\cal I}^U$.

\medskip
\noindent \emph{Case 4.2. $s\in B$}.
In this case, $W_i$ is the node set of a subpath of $P$, say $P'$, consisting of big nodes only. Let $I_{j_1},\ldots, I_{j_p}$ be the intervals corresponding to the vertices in $U_i$, ordered increasingly. The fact that all arcs of $P'$ are in $E_1$ imply that
$V(P') = W_i = \{s_q\mid  1\le q\le p-2k+1\}$ where $s_q = (I_{j_{q}}, \ldots, I_{j_{q+2k-1}})$ and $E(P') = \{(s_q,s_{q+1}) \mid 1\le q\le p-2k\}$.
Moreover, there exists a unique index $q\in \{1,\ldots, p-1\}$ such that
$I^- = I_{j_q}$ and $I^+ = I_{j_{q+1}}$.

Suppose first that $q<k$. Let $(s',s'')$ be the (unique) arc of $P$ such that $s'\not\in W_i$ and $s''\in W_i$.
Then $(s',s'')\in E_0$ and $s''\in B$, and $I_{j_1}<I<I_{j_k}$. Therefore, condition $(4)$ in the definition of an arc in $E_0$
guarantees that interval $I$ intersects at least $k$ intervals from ${\cal I}_{s''}$, and hence also at least $k$ intervals from the set ${\cal I}^U$.

The case when $p-q<k$ is symmetric to the case $q<k$ and can be analyzed using condition $(3)$ in the definition of an arc in $E_0$.

\begin{sloppypar}
Suppose now that $k\le q\le p-k$. Let $\alpha = q-k+1$ and let $s_\alpha = (I_{j_{\alpha}}, I_{j_{\alpha+1}},\ldots, I_{j_{\alpha+2k-1}})$.
Then $s_\alpha\in V(P')$ and $s_\alpha$ is a big node of $D_{k}$. Moreover, $I_{j_{\alpha+k-1}} = I_{j_q} = I^-<I<I^+ = I_{j_{q+1}} = I_{j_{\alpha+k}}$. Therefore, condition $(2)$ in the definition of a big node implies that
$I$ intersects at least $k$ intervals from the set ${\cal I}_{s_\alpha}$, which is a subset of ${\cal I}^{U}$.
\end{sloppypar}

This shows that $U$ is a $k$-dominating set of $G$ and completes the proof.
\end{proof}

\begin{theorem}\label{thm:kd-pig}
For every positive integer $k$, the $k$-domination problem is solvable in time $\mathcal{O}(|V(G)|^{3k})$ in the class of proper interval graphs.
\end{theorem}

\begin{proof}
The digraph $D_k$ has $\mathcal{O}(n^{2k})$ nodes and $\mathcal{O}(n^{4k})$ arcs and can be, together with the length function $\ell$ on its arcs, computed from $G$ directly from the definition in time $\mathcal{O}(n^{4k+1})$. Therefore, Proposition~\ref{prop:correctness-k-domination} implies the the $k$-domination problem is solvable in time $\mathcal{O}(n^{4k+1})$
in the class of $n$-vertex proper interval graphs.
Furthermore, the same speedup as the one used for total $k$-domination in the proof of Theorem~\ref{thm:tkd-pig} applies also to $k$-domination.
This proves Theorem~\ref{thm:kd-pig}.
\end{proof}

\section{The Weighted Problems}\label{sec:kd-weighted}

The approach of Kang et al.~\cite{Kang2017arXiv}, which implies that $k$-domination and total $k$-domination are solvable in time $O(|V(G)|^{6k+4})$ in the class of interval graphs also works for the weighted versions of the problems, where each vertex $u\in V(G)$ is equipped with a non-negative cost $c(u)$ and the task is to find a (usual or total) $k$-dominating set of $G$ of minimum total cost. For both families of problems, our approach can also be easily adapted to the weighed case. Denoting the total cost of a set ${\cal J}$ of vertices (i.e., intervals) by $c({\cal J}) = \sum_{I\in {\cal J}}c(I)$, it suffices to generalize the length function from~\eqref{eq:length}
in a straightforward way, as follows:
\begin{equation*}
\ell(s,s')=\left \{
\begin{array}{ll}
c({\cal I}_{s'}), & \hbox{if $(s,s')\in E_0$ and $s'\neq I_{n+1}$;}\\
c(\min(s')), & \hbox{if $(s,s')\in E_1$;}\\
0, & \hbox{otherwise.}
\end{array}
\right.
\end{equation*}
Except for this change, the algorithms are the same as for the unweighted versions, and the proof of correctness can be adapted easily.
We therefore obtain the following algorithmic result.

\begin{sloppypar}
\begin{theorem}
For every positive integer $k$, the weighted $k$-domination and weighted total $k$-domination problems are solvable in time $\mathcal{O}(|V(G)|^{3k})$ in the class of proper interval graphs.
\end{theorem}
\end{sloppypar}

\section{Conclusion}\label{sec:conclusion}

\begin{sloppypar}
In this work, we developed novel algorithms for weighted $k$-domination and total $k$-domination problems in the class of proper interval graphs. The time complexity was significantly improved, from $\mathcal{O}(n^{6k+4})$ to $\mathcal{O}(n^{3k})$, for each fixed integer $k\ge 1$.
Our work leaves open several questions. Even though polynomial for each fixed $k$, our algorithms are too slow to be of practical use, and the main question is whether the exponential dependency on $k$ of the running time can be avoided. A related question is whether $k$-domination and total $k$-domination problems are fixed-parameter tractable with respect to $k$ in the class of proper interval graphs. Could it be that even the more general problems of {\it vector domination} and {\it total vector domination}~(see, e.g., \cite{MR1741406,MR1961204,MR3027964,MR3497986}), problems which generalize $k$-domination and total $k$-domination when $k$ is part of input,
can be solved in polynomial time for proper interval graphs?
It would also be interesting to determine the complexity of these problems in generalizations of proper interval graphs such as interval graphs, strongly chordal graphs, cocomparability graphs, and AT-free graphs.
\end{sloppypar}

\subsection*{Acknowledgements}

The authors are grateful to Matja\v{z} Krnc for helpful comments.
This work is supported in part by the Slovenian Research Agency (I$0$-$0035$, research program P$1$-$0285$ and research projects N$1$-$0032$, J$1$-$7051$, and J$1$-$9110$). The work for this paper was partly done in the framework of a bilateral project between Argentina and Slovenia, financed by the Slovenian Research Agency (BI-AR/$15$--$17$--$009$) and MINCYT-MHEST (SLO/14/09).


\end{document}